\documentclass[a4paper,11pt]{article} 
\usepackage[english]{babel}
\usepackage[a4paper]{geometry}
\usepackage{amsmath}
\usepackage{amsthm}
\usepackage{amssymb}
\usepackage{color}
\usepackage{hyperref} 
\usepackage{enumerate}

\def\bR{\mathbb{R}}

\def\bT{\mathbb{T}}

\def\bZ{\mathbb{Z}}
\def\cN{\mathcal{N}}
\def\cE{\mathcal{E}}
\def\cK{\mathcal{K}}

\def\ph{\varphi}

\def\cV{\mathcal{V}}
\def\cL{\mathcal{L}}
\def\cR{\mathcal{R}}
\def\cU{\mathcal{U}}
\def\cS{\mathcal{S}}
\def\cD{\mathcal{D}}

\def\wh{\widehat}
\def\be{\begin{equation}}
\def\ee{\end{equation}}

\def\wh{\widehat} 
  
\def\rhs{r.h.s.\ }

\def\eg{e.g.\ }
\def\st{s.t.\ }
\def\hc{\text{h.c.}\ }
\def\vr{V_{\text{ren}}}
\def\PL{ \text{P}_{\text{L}} }

\def\pil{\Pi_{\text{L}} }
\def\pih{\Pi_{\text{H}} }

\newtheorem{theorem}{Theorem}

\newtheorem{lemma}[theorem]{Lemma}

\title{A Short Proof of Bose-Einstein Condensation in the Gross-Pitaevskii Regime and Beyond}

\author{Christian Brennecke\thanks{Institute for Applied Mathematics, University of Bonn, Endenicher Allee 60, 53115 Bonn, Germany} 
\and Morris Brooks\thanks{Institute for Mathematics, University of Zurich, Winterthurerstrasse 190, 8057 Zurich, Switzerland}
\and Cristina Caraci\footnotemark[2]
\and Jakob Oldenburg\footnotemark[2]
}

\begin{document}

\maketitle

\begin{abstract}
We consider dilute Bose gases on the three dimensional unit torus that interact through a pair potential with scattering length of order $ N^{\kappa-1}$, for some $\kappa >0$. For the range $ \kappa\in [0, \frac1{43})$, \cite{ABS} proves complete BEC of low energy states into the zero momentum mode based on a unitary renormalization through operator exponentials that are quartic in creation and annihilation operators. In this paper, we give a new and self-contained proof of BEC of the ground state for $ \kappa\in [0, \frac1{20})$ by combining some of the key ideas of \cite{ABS} with the novel diagonalization approach introduced recently in \cite{Br}, which is based on the Schur complement formula. In particular, our proof avoids the use of operator exponentials and is significantly simpler than \cite{ABS}. 
\end{abstract}

\section{Introduction and Main Result}\label{sec:intro}
We consider $N$ interacting bosons in $\Lambda:= \bT^3 = \bR^3/ \bZ^3$ with Hamiltonian 
        \be\label{eq:HN} H_N = \sum_{i=1}^N -\Delta_{x_i} + \sum_{1\leq i<j\leq N} N^{2-2\kappa} V (N^{1-\kappa}(x_i -x_j)), \ee
acting in $L^2_s (\Lambda^N)$, the Hilbert space consisting of functions in $L^2 (\Lambda^N)$ that are invariant with respect to permutations of the $N$ particles. We assume the interaction potential $V \in L^{1}(\bR^3)$ to have compact support, to be radial and to be pointwise non-negative.  

Note that analyzing $H_N$ is equivalent to analyzing the Hamiltonian of $N$ bosons interacting through the unscaled potential $V$ in $ \bR^3/ L\bZ^3$ for $L=N^{1-\kappa}$. In this sense, we consider regimes of strongly diluted systems of bosons with number of particles density $ N^{3\kappa-2} \ll 1 $ (as long as $ \kappa < \frac23$). The case $\kappa=0 $ corresponds to the Gross-Pitaevskii (GP) regime and the case $\kappa =\frac23$ corresponds to the usual thermodynamic limit (with number of particles density equal to one). 

In this paper, we are interested in understanding low energy properties of the Bose gas in regimes that interpolate between the GP and thermodynamic limits. Based on \cite{D, LY}, it is well-known that the ground state energy $ E_N := \inf \text{spec}(H_N)$ is equal to
        \[E_N = 4\pi \mathfrak{a}N^{1+\kappa} +o(N^{1+\kappa}), \]
where $ \mathfrak{a}$ denotes the scattering length of the potential $V$ and where $o(N^{1+\kappa})$ denotes an error of subleading order, that is $ \lim_{N\to\infty} o(N^{1+\kappa})/N^{1+\kappa}=0$. Recall that under our assumptions the scattering length of $V$ is characterized by
        \[\begin{split}
        8\pi \mathfrak{a} &=   \inf \bigg\{ \int_{\bR^3}dx\,\Big( 2|\nabla f(x)|^2+V(x)|f(x)|^2\Big): \lim_{|x|\to\infty}f(x)=1 \bigg\}. 
        \end{split}\]
A question closely related to the computation of the ground state energy is whether the ground state exhibits Bose-Einstein condensation (BEC). If $\psi_N$ denotes the ground state vector, this means that the largest eigenvalue of the associated reduced one particle density matrix $ \gamma_N^{(1)} = \text{tr}_{2,\dots,N} |\psi_N\rangle\langle\psi_N|$ remains of size one in the limit $N\to\infty$:
        \[\liminf_{N\to\infty}\|\gamma_N^{(1)}\|_{\text{op}}>0.\]
Proving BEC in the thermodynamic limit is a difficult open problem in mathematical physics. For strongly diluted systems, on the other hand, there has recently been great progress in proving that low energy states exhibit BEC. The first proof of BEC has been obtained in \cite{LS1} in the GP regime\footnote{It is worth to point out that the arguments of \cite{LS1}, which build on energy bounds from \cite{LY, LSY}, can in fact be used to prove BEC in the parameter range $\kappa\in[0,\frac1{10})$; see also \cite[Chapter 7]{LSSY}.}, implying that for $ \ph_0:=1_{|\Lambda}\in L^2(\Lambda)$ one has that
        \be \label{eq:BEC} \lim_{N\to\infty}   \langle\ph_0,\gamma_N^{(1)} \ph_0\rangle=1. \ee
This result has later been extended to approximate ground states in \cite{LS2, NRS} and the works \cite{BBCS1, BBCS4} have proved \eqref{eq:BEC} with the optimal rate of convergence. Since then, several generalizations and simplified proofs have been obtained in \cite{NNRT, ABS, F, H, BSS1, CCS1, BocSei, NRT}. Notice that such results can be used to derive the low energy excitation spectrum of $H_N$ in accordance with Bogoliubov theory \cite{Bog}, see \eg \cite{BBCS3, BBCS4, NT, BSS2, HST, BCaS, CCS2, BCOPS, Br, COSAS}. 

In recent years, progress has also been made in regimes that interpolate between the GP and thermodynamic limits. Based on unitary renormalizations developed first in the dynamical context \cite{BDS, BS} and in the context of the derivation of the excitation spectrum in the GP regime \cite{BBCS2, BBCS3}, the work \cite{ABS} proves BEC for approximate ground states in regimes $ \kappa \in [0,\frac1{43})$. A different method, that is based on box localization arguments, has been introduced in \cite{F} which proves BEC in the larger parameter range $ \kappa \in [0, \frac25+\epsilon)$, for some sufficiently small $\epsilon>0$. This result represents currently the best available parameter range and it is closely tied to the computation of the second order correction to the ground state energy, which turns out to be of order $ N^{ 5\kappa/2}$ \cite{YY, FS1,BCS, FS2, HHNST}.  

The methods introduced in \cite{ABS} and \cite{F} have both certain advantages. While \cite{F} obtains the currently best parameter range and applies to a large class of potentials including hard-core interactions, it is based on box localization arguments and therefore involves the change of boundary conditions\footnote{To be more precise, the localization procedure of \cite{F} replaces the standard Laplacian in the periodic setting by a more involved localized kinetic energy operator, see \cite[Eq. (2.7)]{F}. For a recent overview that focuses on the key steps of the energy bounds in the simpler translation invariant setting, see \cite{FGJMO}.}. This makes the derivation of suitable lower bounds more complicated, compared to the translation invariant setting, and essentially restricts the method to obtaining lower bounds while upper bounds require separate tools. The method of \cite{ABS}, on the other hand, does not require localization and enables both upper and lower bounds at the same time. However, it only applies to soft potentials satisfying some mild integrability assumption. Moreover, controlling the error terms in the operator expansions quickly becomes rather challenging and this is among the main reasons why the method only works in a much more restricted parameter range. 

In this paper, our goal is to revisit the strategy of \cite{ABS}. However, instead of renormalizing the system through unitary conjugations by quartic operator exponentials, we proceed as in \cite{Br} whose renormalization is based on the Schur complement formula applied to the two body problem and on lifting it in a suitable sense to the $N$ body setting. As a consequence, our proof becomes significantly simpler and shorter compared to the one in \cite{ABS}. Although our results are still only valid in a small parameter range compared to \cite{F}, our arguments are elementary, self-contained and do neither require box localization methods nor operator exponential expansions. 

\begin{theorem}\label{thm:main}
Let $H_N$ be defined as in \eqref{eq:HN} for $ \kappa\in [0, \frac1{20})$ and denote by $ \gamma_N^{(1)}$ the one particle reduced density associated to its normalized ground state vector $\psi_N$. Then 
        \[\lim_{N\to\infty}   \langle\ph_0,\gamma_N^{(1)} \ph_0\rangle=1. \]
\end{theorem}
\vspace{0.1cm}
\noindent \textbf{Remarks:} 
\begin{enumerate}[1)]
\item Theorem \ref{thm:main} applies to the ground state vector $\psi_N$ of $H_N$. With some additional effort that involves the use of number of particles localization arguments, we expect that our results could also be proved for approximate ground states $ \phi_N$ that satisfy $ \langle \phi_N, H_N\phi_N\rangle \leq 4\pi \mathfrak{a}N^{1+\kappa}+ o(N)$. To keep our arguments as short and simple as possible, we omit the details and focus on the ground state vector $\psi_N$. 
\item In our proof of Theorem \ref{thm:main}, we assume the relatively mild a priori information that the ground state energy $ E_N $ is bounded from above by $ E_N \leq 4\pi \mathfrak{a}N^{1+\kappa}+ o(N)$, if $\kappa<\frac1{20}$. Based on ideas similar to those presented below, this could be proved with little additional effort in a self-contained way. Since this has already been explained in \cite{Br} (which obtains a more precise upper bound on $ E_N$ for all $\kappa<\frac2{13}$ based on the evaluation of the energy of suitable trial states, see \cite[Theorem 3]{Br}), however, we refer the interested reader to \cite{Br} for the details.
\end{enumerate}

\section{Proof of Theorem \ref{thm:main} }
In the following, let us denote by $ a_k$ and $a^*_k$ the annihilation and, respectively, creation operators associated with the plane waves $ x\mapsto \ph_k(x) :=e^{ikx} \in L^2(\Lambda)$ of momentum $k$, for $k\in \Lambda^*:=2\pi\bZ^3$. They satisfy the canonical commutation relations $ [a_p, a_q^*] = \delta_{p,q}$ and $ [a_p,a_q]=[a_p^*, a_q^*]=0$, and they can be used to express $ H_N$  as 
        \[ H_N = \sum_{r\in \Lambda^*}|r|^2 a^*_ra_r + \frac{N^{\kappa}}{2N}\sum_{p,q,r\in \Lambda^*}  \widehat V(r/N^{1-\kappa}) a^*_{p+r}a^*_{q-r}a_pa_q, \]
where $ \widehat V(r) = \int_{\bR^3}dx\, e^{-irx}V(x)$ denotes the standard Fourier transform of $V$. 

Now, denote by $ V_N$ the two body operator that multiplies by $N^{2-2\kappa}V(N^{1-\kappa}(x_1-x_2)) $ in $ L^2(\Lambda^2)$ and define for $\alpha\in [0, 1-\kappa]$ the low momentum set 
        \be\label{eq:defPL}\PL := \big\{p\in\Lambda^*: |p|\leq N^{\alpha} \big\}. \ee
Denote, moreover, by $\pil:L^2(\Lambda^2)\to L^2(\Lambda^2)$ the orthogonal projection onto $$ \overline { \text{span}(\ph_k\otimes \ph_l:k,l\in \PL)}  $$ and set $ \pih:=1-\pil$. Then, as explained in detail in \cite{Br}, a straightforward application of the Schur complement formula implies the many body lower bound 
        \be \begin{split}\label{eq:lwrbnd}   
        H_N \geq \sum_{r\in \Lambda_+^*}|r|^2 c^*_r c_r + \frac{N^{\kappa}}{2N}\!\!\!\!\sum_{\substack{p,q,r\in\Lambda^*: \\ p,q, p+r, q-r \in \PL}}\!\!\!\!\! \langle \ph_{p+r}\otimes \ph_{q-r}, \vr \ph_p\otimes \ph_q  \rangle   a^*_{p+r}a^*_{q-r}a_pa_q - R_N, 
        \end{split} \ee
where we set $ \Lambda_+^* := \Lambda^*\setminus\{0\}$ as well as
        \be\begin{split}\label{eq:defs1} 
        c_r &:= a_r + \frac{N^{\kappa}}{N}\sum_{\substack{(p,q)\in\PL^2  }}   \langle \ph_{p+q-r}\otimes \ph_{r}, \eta \,  \ph_p\otimes \ph_q  \rangle   a^*_{p+q-r}a_pa_q,\\
        \eta& := N^{1-\kappa}\,\pih \big[ \pih(-\Delta_{x_1}-\Delta_{x_2} + V_N )\pih \big]^{-1}\pih V_N\pil   ,\\
        \vr & := N^{1-\kappa}\big( V_N - V_N \pih \big[ \pih(-\Delta_{x_1}-\Delta_{x_2} + V_N )\pih \big]^{-1}\pih V_N \big), 
        \end{split}\ee
and where the three body error term $R_N$ is given by
         \be\begin{split}\label{eq:defs2} 
        R_N& := \frac{N^{2\kappa}}{N^2}\!\!\!\sum_{\substack{r, p,q, s, t\in\Lambda^* }} \!\!\! |r|^2 \langle \eta \,\ph_p\otimes \ph_q, \ph_{p+q-r}\otimes \ph_{r}    \rangle  \langle \ph_{s+t-r}\otimes \ph_{r}, \eta \,  \ph_{s}\otimes \ph_{t}  \rangle  \\
        &\hspace{3cm}\times     a_p^*a^*_q  a^*_{s+t-r}a_{p+q-r}a_{s}a_{t} . 
        \end{split}\ee
Notice that we used that both $\eta$ and $\vr$ preserve the total momentum in $ L^2(\Lambda^2)$. 

Let us briefly comment on the main ideas leading to \eqref{eq:lwrbnd}. Viewing $ V_N = \pil V_N \pil + (\pih V_N \pil+\text{h.c.}) + \pih V_N\pih$ and hence the Hamiltonian $\mathcal{H}_2:=-\Delta_{x_1}-\Delta_{x_2}+V_N$ of the two body problem as a block matrix, one can block-diagonalize the latter using the Schur complement formula. This renormalizes the low-momentum interaction to $N^{\kappa-1} \pil\vr\pil$ while the large momentum interaction $\pih V_N\pih $ is left untouched. The (non-symmetric) map that block-diagonalizes $\mathcal{H}_2$ is of the form $ \cS_\eta= 1+N^{\kappa-1}\eta$ and, in order to obtain an analogous renormalization of the many body interaction, it seems natural to lift $\cS_\eta$ to the unitary generalized Bogoliubov transformation 
        \[\begin{split}
        & \cU_\eta := \exp(\cD_\eta-\cD_\eta^*)\, \big(\approx 1+ \cD_\eta-\cD_\eta^*\big), \text{ where }\\
           & \cD_\eta:= \frac{N^{\kappa}}{2N}\sum_{\substack{ p,q,r\in\Lambda^*: (p,q)\in \PL^2,\\ (p-r,q+r)\in (\PL^2)^c}}\langle \ph_{p-r}\otimes \ph_{q+r}, \eta \,  \ph_p\otimes \ph_q  \rangle   a^*_{p-r}a^*_{q+r} a_pa_q.
        \end{split} \]
On a conceptual level, this approach corresponds to the one pursued in \cite{ABS} (in particular, the role of $\eta$ defined in \eqref{eq:defs1} is similar to that of $\eta_{\text H}$ defined in \cite{ABS} through the zero energy scattering equation). Compared to that, a key idea of \cite{Br} is to expand $H_N$ directly around powers of suitably modified creation and annihilation operators, including \eg $ c_r = a_r + [a_r, \cD_\eta] \,(\approx \cU_\eta^* a_r\,\cU_\eta $). This leads to the low momentum renormalization of the many body interaction in a simple way and avoids the use of operator exponential expansions. Notice that this approach is reminiscent of  previously introduced ideas in \cite{BFS, FS1}. Finally, let us stress that, although the bound \eqref{eq:lwrbnd} is all we need in view of Theorem \ref{thm:main}, \cite{Br} derives in fact exact algebraic identities. Similarly as in \cite{ABS}, what is dropped in \eqref{eq:lwrbnd} is the non-renormalized high momentum part of the potential energy. 

Proceeding as in \cite[Lemma 1]{Br}, let us record the useful upper bounds
        \be \label{eq:bndvr} \begin{split}
        \big|\langle \ph_{k_1}\otimes \ph_{k_2}, \vr \ph_{k_3}\otimes \ph_{k_4}  \rangle \big|&\leq C, \\
        \big|\langle \ph_{k_1}\otimes \ph_{k_2}, \vr \ph_{k_3}\otimes \ph_{k_4} \rangle   -   8\pi \mathfrak{a} \big| \! &\leq   C  N^{\kappa -1
        }    \Big(N^{\alpha}  +  \sum_{i=1}^4 N^{-\alpha} |k_i|^2\Big),
        \end{split}\ee
for all $ k_1,k_2,k_3,k_4\in \Lambda^*$ satisfying $k_1+k_2=k_3+k_4$ and $\langle \ph_{k_1}\otimes \ph_{k_2}, \vr \ph_{k_3}\otimes \ph_{k_4}  \rangle=0$ in case $k_1+k_2\neq k_3+k_4$. The bounds \eqref{eq:bndvr} imply in particular that   
        \be \label{eq:bndeta} 
        \big|\langle \ph_{k_1}\otimes \ph_{k_2}, \eta \,  \ph_{k_3}\otimes \ph_{k_4}  \rangle\big|\leq  \frac{C\,\delta_{k_1+k_2,k_3+k_4}}{|k_1|^2+ |k_2|^2} \textbf{1}_{( \PL^2)^c}((k_1,k_2))\textbf{1}_{ \PL^2}((k_3,k_4)).
        \ee 
For completeness, we prove \eqref{eq:bndvr} and \eqref{eq:bndeta} in Appendix \ref{appx}, following \cite[Appendix A]{Br}.

Based on \eqref{eq:lwrbnd}, \eqref{eq:bndvr} and \eqref{eq:bndeta}, the proof of Theorem \ref{thm:main} follows by carefully estimating the three terms on the \rhs in \eqref{eq:lwrbnd} and by combining these estimates with some mild a priori information on the ground state energy. Before summarizing the key steps, let us introduce the following additional notation: for every $\zeta\geq 0$, we set
        \[ \cN_{>\zeta} := \sum_{r\in\Lambda^*: |r|>\zeta} a^*_ra_r \]
and similarly, we define $ \cN_{\geq \zeta}, \cN_{<\zeta}$ and $\cN_{\leq\zeta}$. Moreover, we set $\cN :=\cN_{\geq 0}\ (\equiv N) $, $ \cN_+:=\cN_{>0}$ and $ \cK := \sum_{r\in\Lambda^*_+}|r|^2a^*_ra_r$. It is an elementary observation that 
        \[\label{eq:BECcrit}1-\langle \ph_0, \gamma_N^{(1)}\ph_0\rangle = N^{-1} \langle\psi_N ,\cN_+\psi_N\rangle.\]
Equipped with the previous identity, the key of our proof is to derive a coercivity bound
        \[ H_N \geq 4\pi \mathfrak{a} N^{1+\kappa} + c\, \cN_+ + \cE   \]
for some constant $c>0$ and some error $\cE$ which is of size $o(N)$ in the ground state $\psi_N$. The number of excitations $\cN_+$ is extracted from the modified kinetic energy operator in \eqref{eq:lwrbnd} (the first term on the \rhs in \eqref{eq:lwrbnd}) while the leading order energy $4\pi \mathfrak{a} N^{1+\kappa}$ is extracted from the renormalized potential energy (the second term on the \rhs in \eqref{eq:lwrbnd}). This is explained in Lemmas \ref{lm:K} and \ref{lm:Vren} which represent the key of the whole argument. 

The error terms, on the other hand, turn all out to be related to the number of excitations with large momenta. Following \cite{ABS}, the key tool we use below to control such errors is a simple Markov bound combined with the trivial fact that $E_N\leq C N^{1+\kappa}$:
        \be\label{eq:markov} \cN_{ > N^{\beta}}\leq N^{-2\beta}\cK \leq N^{-2\beta}H_N.  \ee
In particular $ \langle \psi_N, \cN_{ > N^{\beta}}\psi_N \rangle\leq C N^{1+\kappa-2\beta} = o(N)$ as soon as $ 2\beta >\kappa$, if $\psi_N$ denotes an approximate ground state vector. In Lemma \ref{lm:apriori}, we slightly generalize the bound \eqref{eq:markov} to products of the kinetic energy with number of particles operators for large momenta. 

\begin{lemma}\label{lm:K}
Suppose $\delta \in (\frac{\kappa}{2},\alpha)$, then we have that
        \[\sum_{r\in \Lambda_+^*}|r|^2 c^*_r c_r \geq 4\pi^2 \big(\cN_{< N^{ \delta}} -a^*_0a_0 \big) + \cE_\delta \]
for a self-adjoint operator $\cE_\delta$ which satisfies for some $C>0$ and $N$ large enough that
        \[
         \pm \cE_{\delta} 
         \leq  CN^{\kappa+\frac{\delta}{2}-\frac{3\alpha}{2}-1}
          (\cK + N )\cN_{> N^\alpha/3}.
        \]
\end{lemma}
\begin{proof}
Recalling the definition of $c_r$ in \eqref{eq:defs1} and setting
        \[d_r = \frac{N^{\kappa}}{N}\sum_{\substack{(p,q)\in\PL^2  }}   \langle \ph_{p+q-r}\otimes \ph_{r}, \eta \,  \ph_p\otimes \ph_q  \rangle   a^*_{p+q-r}a_pa_q,\]
so that $c_r = a_r + d_r$, we lower bound 
        \[\begin{split}
         &\sum_{r\in \Lambda_+^*}|r|^2 c^*_r c_r
         -4\pi^2 \big(\cN_{< N^{ \delta}} -a^*_0a_0 \big)\\
         &\geq \sum_{r\in \Lambda_+^*:0 < |r| < N^\delta}4\pi^2 c^*_r c_r
         -\sum_{r\in \Lambda_+^*:0 < |r| < N^\delta} 4\pi^2a^*_r a_r\\
          &= \sum_{r\in \Lambda_+^*:0 < |r| < N^\delta} 4\pi^2\big(   d^*_ra_r +a^*_r d_r + d^*_r d_r\big)\\
         &\;\geq   \frac{4\pi^2N^{\kappa}}{N} \!\!\!\!\!\!\!\sum_{p,q, r\in\Lambda^*: 0 < |r| < N^\delta }   \!\!\!  \langle \ph_{p+q-r}\otimes \ph_{r}, \eta \,  \ph_p\otimes \ph_q  \rangle  a^*_r a^*_{p+q-r}a_pa_q +\text{h.c.},  
        \end{split} \]
where in the first and last steps, we used the positivity of $c^*_rc_r\geq0 $ and $d^*_rd_r\geq 0$, respectively. With the bound \eqref{eq:bndeta} and Cauchy-Schwarz, we then obtain for $ \xi\in L^2_s(\Lambda^N)$ 
        	\[\begin{split}
        	& \bigg|  N^{\kappa-1}  \sum_{p,q, r\in\Lambda^*: 0 < |r| < N^\delta }      \langle \ph_{p+q-r}\otimes \ph_{r}, \eta \,  \ph_p\otimes \ph_q  \rangle \langle\xi a^*_r a^*_{p+q-r}a_pa_q \xi\rangle \bigg| \\
        	&\; \leq  C N^{\kappa-2\alpha-1}\!\!\!\sum_{\substack{(p,q,r)\in P_L^3: \\ 0 < |r| < N^\delta,\, |p|> N^\alpha/3,\, p+q-r\in  \PL^c }} \frac{|r|}{|q|+1}\| a_r a_{p+q-r}\xi\|    \frac{|q|+1}{|r|}\|a_p a_q \xi\|\\
         & \; \leq CN^{\kappa+\frac{\delta}{2}-\frac{3\alpha}{2}-1}
          \langle \xi, (\cK + N )\cN_{> N^\alpha/3}\xi\rangle.
        \end{split}\]
Notice that due to the constraint $ p+q-r \in \PL^c $ and the condition $|r| < N^\delta$ for $\delta < \alpha$, at least one of the momenta $p$ and $q$ has to be larger than $N^\alpha/3$ for large $N$.
\end{proof}

\begin{lemma}\label{lm:Vren}
There exists a constant $C>0$ such that 
        \begin{equation}\begin{split}
        \label{eq:Vren}&\frac{N^{\kappa}}{2N}\!\!\!\sum_{\substack{p,q,r\in\Lambda^*: \\ p,q, p+r, q-r \in \PL}}\!\!\!\!\! \langle \ph_{p+r}\otimes \ph_{q-r}, V_{\emph{ren}}  \ph_p\otimes \ph_q  \rangle   a^*_{p+r}a^*_{q-r}a_pa_q \\
        &\;\geq 4\pi \mathfrak{a} N^{1+\kappa} - CN^\kappa\cN_{>N^\alpha} -CN^{\kappa +3\alpha} - C N^{2\kappa +2\alpha-1} (\cK + N)
        \end{split}\end{equation}
\end{lemma}
\begin{proof}
We use the bound \eqref{eq:bndvr} together with the fact that $|p|, |q|, |r|\leq 2 N^\alpha$ if $p,q,p-r,q+r\in \PL $ to replace $\vr$ as follows: for every $\xi\in L^2_s(\Lambda^N)$, we have that
        \[\begin{split}
        &\frac{N^{\kappa}}{2N}\!\sum_{\substack{p,q,r \in \Lambda^*:\\
        p, q, p+r, q-r \in P_L}}\! 
        \big| \langle \ph_{p+r}\otimes \ph_{q-r}, \vr \ph_p\otimes \ph_q     \rangle -8\pi \mathfrak{a}\big|  |\langle \xi, a^*_{p+r}a^*_{q-r}a_pa_q\xi\rangle|\\
        &\leq C N^{2\kappa + \alpha - 2}\!\sum_{\substack{p,q,r \in \Lambda^*:\\
        p, q, p+r, q-r \in P_L}}\!   
        \frac{|p+r|+1}{|p|+1}\|a_{p+r}a_{q-r}\xi\|
        \frac{|p|+1}{|p+r|+1}\|a_pa_q\xi\|\\
        &\leq C N^{2\kappa +2\alpha-1} \langle \xi,  (\cK + N) \xi \rangle.
        \end{split}\]
As a consequence, we get the lower bound 
        \[\begin{split}
         & \frac{N^{\kappa}}{2N}\!\!\!\sum_{\substack{p,q,r\in\Lambda^*: \\ p,q, p+r, q-r \in \PL}}\!\!\!\!\! \langle \ph_{p+r}\otimes \ph_{q-r}, V_{\emph{ren}}  \ph_p\otimes \ph_q  \rangle   a^*_{p+r}a^*_{q-r}a_pa_q \\
        &\geq \frac{4\pi \mathfrak{a} N^{\kappa}}{N}\!\!\!\sum_{\substack{p,q,r\in\Lambda^*: \\ p,q, p+r, q-r \in \PL}}\!\!\!\!\!   a^*_{p+r}a^*_{q-r}a_pa_q -  C N^{2\kappa +2\alpha-1}  (\cK + N). 
        \end{split}\]
The lemma now follows by combining this estimate with the lower bound
        \[\begin{split}
           &\frac{4\pi \mathfrak{a} N^{\kappa}}{N}\!\!\!\sum_{\substack{p,q,r\in\Lambda^*: \\ p,q, p+r, q-r \in \PL}}\!\!\!\!\!   a^*_{p+r}a^*_{q-r}a_pa_q  \\
           & \;= \frac{4\pi \mathfrak{a} N^{\kappa}}{N} \sum_{r\in\Lambda^*} \bigg(\sum_{\substack{ q\in\PL:\\q+r\in \PL}} a^*_{q}a_{q+r}\bigg)^* \bigg(\sum_{\substack{ q\in\PL:\\q+r\in \PL}} a^*_{q}a_{q+r}\bigg) - \frac{4\pi \mathfrak{a} N^{\kappa}}{N}\!\!\!\sum_{\substack{p,r\in\Lambda^*: \\ p, p+r \in \PL}}\!\!\!\!\!   a^*_{p+r} a_{p+r}\\
            &\geq \frac{4\pi \mathfrak{a} N^{\kappa}}{N}   \bigg(\sum_{\substack{ q\in\PL}} a^*_{q}a_{q}\bigg)^* \bigg(\sum_{\substack{ q\in\PL}} a^*_{q}a_{q}\bigg) - \frac{4\pi \mathfrak{a} N^{\kappa}}{N}\!\!\!\sum_{\substack{p,r\in\Lambda^*: \\ p, p+r \in \PL}}\!\!\!\!\!   a^*_{p+r} a_{p+r}\\
           &\;=  \frac{4\pi \mathfrak{a} N^{\kappa}}{N}\big(N- \cN_{>N^{\alpha}}\big)^2- \frac{4\pi \mathfrak{a} N^{\kappa}}{N}\!\!\!\sum_{\substack{p,r\in\Lambda^*: \\ p, p+r \in \PL}}\!\!\!\!\!   a^*_{p+r} a_{p+r}\\
           &\;\geq 4\pi \mathfrak{a} N^{1+\kappa} - 8\pi \mathfrak{a} N^{\kappa}\cN_{>N^{\alpha}} - C N^{\kappa+3\alpha},
        \end{split}\]    
      where in the last step we dropped the positive contribution proportional to $\cN_{>N^{\alpha}}^2$ and where we used that $ \cN_{\leq N^{\alpha}}\leq N$ as well as $ |\PL|\leq CN^{3\alpha}$.  
\end{proof}

\begin{lemma}\label{lm:RN}
Let $R_N$ be as in \eqref{eq:defs2} and let $0\leq \beta < \alpha$. Then, there exists $C>0$ such that for $N$ large enough, we have that
\[\begin{split}
        \pm R_N 
        &\leq CN^{2\kappa-2\alpha-2}
        \Big(N^{4\alpha} (\cN_{>N^{\beta}}+N^{3\alpha})
        + N^{\frac52\alpha + \frac32\beta +\frac12} (\cN_{>N^{\beta}}+N^{3\alpha})^{\frac12}
        +N^{\frac32\alpha+\frac52\beta+1}\Big)\\
        &\hspace{2.05cm}\times\big(\cK + N + N^{5\beta}\big) \big(\cN_{>N^{\alpha}/3}+1\big). \qedhere
        \end{split}\]
\end{lemma}
\begin{proof}
Given $ \xi\in L^2_s(\Lambda^N)$, we apply the bound \eqref{eq:bndeta} to get
        \[\begin{split}
            |\langle\xi, R_N\xi\rangle|&\leq \frac{CN^{2\kappa}}{N^2}\!\!\!\sum_{\substack{r\in\Lambda^*_+, p,q, s, t\in\Lambda^*:\\ (p+q-r,r),(s+t-r,r)\in (\PL^2)^c,\\ (p,q),(s,t)\in \PL^2 }} \!\!\!  \frac{|r|^2 | \langle \xi, a_p^* a_q^*  a_{s+t-r}^* a_{p+q-r}a_{s}a_{t}\xi\rangle|}{(|p+q-r|^2+ |r|^2)(|s+t-r|^2+ |r|^2)}\\
            &\leq  CN^{2\kappa-2\alpha-2} \!\!\!\sum_{\substack{r\in\Lambda^*_+, p,q, s, t\in\Lambda^*:\\ (p+q-r,r),(s+t-r,r)\in (\PL^2)^c,\\ (p,q),(s,t)\in\PL^2 }} \!\!\!   |\langle \xi, a_p^* a_q^*  a_{s+t-r}^* a_{p+q-r}a_{s}a_{t}\xi\rangle|.
        \end{split}\]
In order to control the sum on the right hand side, we split it according to two types of restrictions: first, consider another scale $N^{\beta}$, for $\beta < \alpha$, and consider the cases in which the momenta
$p,q,s,t\in \PL^4$ are smaller or greater than $N^\beta$. We consider the cases 
\be\label{eq:cases1}\begin{split} 
&(1)\;\;|p|, |q|, |s|, |t| \leq N^\beta,\\
    &(3)\;\;|p|, |q| > N^\beta \text{ and } |s|, |t| \leq N^\beta  , \\
    &(5)\;\;|p|, |q|, |s| > N^\beta \text{ and } |t| \leq N^\beta  ,
\end{split}
\hspace{1cm}
\begin{split}
     &(2)\;\;|p| > N^\beta \text{ and } |q|, |s|, |t| \leq N^\beta,\\
    &(4)\;\;|p|, |s| > N^\beta \text{ and } |q|, |t| \leq N^\beta,\\
    &(6)\;\; |p|, |q|, |s|, |t| > N^\beta.  
\end{split}\ee
Furthermore, the conditions $(p+q-r,r), (s+t-r,r)\in (\PL^2)^c$ imply that at least one of  $p, q, p+q-r$ and one of $ s,t, s+t-r$ is greater than $N^\alpha /3$: we consider the cases
\be\label{eq:cases2}\begin{split}
    &(a)\;\;|p+q-r|, |s+t-r|> N^\alpha/3,  \\
    &(c)\;\; |p+q-r|, |s|> N^\alpha/3 ,
\end{split}
\hspace{1cm}
\begin{split}
     &(b)\;\;|p|, |s+t-r|> N^\alpha/3,\\
    & (d)\;\;|p|, |s|> N^\alpha/3. 
\end{split}\ee
Now, using symmetries among and within the pairs $(p,q)\in\PL^2$ and $(s,t)\in\PL^2$, one readily sees that for N large enough, such that $N^\beta < N^\alpha/3$, we have that
        \[\begin{split}
           & \sum_{\substack{r\in\Lambda^*_+, p,q, s, t\in\Lambda^*:\\ (p+q-r,r),(s+t-r,r)\in (\PL^2)^c,\\ (p,q),(s,t)\in\PL^2 }} \!\!\!   |\langle \xi, a_p^* a_q^*  a_{s+t-r}^* a_{p+q-r}a_{s}a_{t}\xi\rangle|\\
           &\;\leq C \bigg( \sum_{j=1}^6\Sigma_{ja}(\xi) + \sum_{j=2}^6\Sigma_{jb}(\xi) + \Sigma_{5c}(\xi) +\sum_{j=4}^6\Sigma_{jd}(\xi)\bigg),
        \end{split}\]
where $ \Sigma_{j\alpha}$, for $j\in \{1,\ldots, 6\}$ and $ \alpha\in\{a,b,c,d\}$, refers to the contribution 
        \[\Sigma_{j\alpha}(\cdot):=\sum_{\substack{r\in\Lambda^*_+, p,q, s, t\in\Lambda^*: p,q,s,t,\\ p+q-r,s+t-r \text{ satisfy } j) \text{ and } \alpha) }} \!\!\!   |\langle \cdot, a_p^* a_q^*  a_{s+t-r}^* a_{p+q-r}a_{s}a_{t}\cdot\rangle|\geq 0. \]
Here, the restriction labels $j\in \{1,\ldots, 6\}$ and  $\alpha\in\{a,b,c,d\}$ refer to \eqref{eq:cases1} and \eqref{eq:cases2}, respectively. Applying basic Cauchy-Schwarz estimates as in Lemmas \ref{lm:K} and \ref{lm:Vren}, we find
        \[\begin{split}
        \Sigma_{1a} &\leq C N^{4\beta+1} (\cK + N) (\cN_{>N^{\alpha}/3}+1),\\
        \Sigma_{2a},\Sigma_{3a}&\leq CN^{2\alpha + 2\beta +\frac12} (\cK + N) (\cN_{>N^{\alpha}/3}+1) (\cN_{>N^{\beta}}+1)^{\frac12},\\
        \Sigma_{2b} &\leq CN^{\frac32\alpha+\frac52\beta+1} (\cK + N) (\cN_{>N^{\alpha}/3}+1),\\
        \Sigma_{3b} &\leq C\big( N^{\frac32\alpha + \frac52\beta +\frac12}  (\cN_{>N^{\beta}}+N^{3\alpha})^{\frac12}  + N^{3\alpha+\beta}   (\cN_{>N^{\beta}}+N^{3\alpha})\big)\\
        &\hspace{0.5cm}\times(\cK + N + N^{5\beta}) (\cN_{>N^{\alpha}/3}+1),\\
        \Sigma_{4a},\Sigma_{5a},\Sigma_{6a}  &\leq C N^{4\alpha} (\cK + N) (\cN_{>N^{\alpha}/3}+1) (\cN_{>N^{\beta}}+1),\\
        \Sigma_{4b},\Sigma_{5b},\Sigma_{6b} 
        &\leq C\big( N^{\frac52\alpha + \frac32\beta +\frac12}  (\cN_{>N^{\beta}}+1)^{\frac12}  + N^{4\alpha}   (\cN_{>N^{\beta}}+1)\big) (\cK + N) (\cN_{>N^{\alpha}/3}+1), \\
        \Sigma_{5c} &\leq  C\big( N^{2\alpha + 2\beta +\frac12}   (\cN_{>N^{\beta}}+1)^{\frac12} +N^{\frac72\alpha +\frac12\beta}  (\cN_{>N^{\beta}}+1)\big) (\cK + N) (\cN_{>N^{\alpha}/3}+1), \\
        \Sigma_{4d}, \Sigma_{5d}, \Sigma_{6d} &\leq C\big( N^{\frac52\alpha + \frac32\beta +\frac12}   (\cN_{>N^{\beta}}+1)^{\frac12} +N^{4\alpha}   (\cN_{>N^{\beta}}+1) +N^{\alpha + 3\beta+1}\big) \\
        &\hspace{0.5cm}\times(\cK + N)(\cN_{>N^{\alpha}/3}+1)  .
        \end{split}\]
Here, an inequality of the form $ \Sigma_{j\alpha} \leq \cL $ for a non-negative self-adjoint operator $\cL $ refers to the statement that $ \Sigma_{j\alpha}(\xi)\leq \langle\xi,\cL\,\xi\rangle $, for all $\xi\in L^2_s(\Lambda^N)$. 
In order to illustrate more explicitly how to bound the above terms, consider for example $\Sigma_{1a}$: here we bound 
\[
\begin{split}
    \Sigma_{1a}
    &\leq \sum_{\substack{r\in\Lambda^*_+, p,q, s, t\in\Lambda^*:\\ |p|,|q|,|s|,|t|\leq N^\beta\\
    |p+q-r|, |s+t-r|> N^\alpha/3}} \!\!\!   \| a_p a_q  a_{s+t-r}\cdot\| \|a_{p+q-r}a_{s}a_{t}\cdot\|\\
    &\leq \Big(\sum_{\substack{r\in\Lambda^*_+, p,q, s, t\in\Lambda^*:\\ |p|,|q|,|s|,|t|\leq N^\beta\\
    |p+q-r|, |s+t-r|> N^\alpha/3}} \!\!\!  
    \Big(\frac{|p|+1}{|s|+1}\Big)^2 \| a_p a_q  a_{s+t-r}\cdot\|^2 \Big)^{1/2}\\
    &\quad\times\Big(\sum_{\substack{r\in\Lambda^*_+, p,q, s, t\in\Lambda^*:\\ |p|,|q|,|s|,|t|\leq N^\beta\\
    |p+q-r|, |s+t-r|> N^\alpha/3}} \!\!\!  
    \Big(\frac{|s|+1}{|p|+1}\Big)^2\|a_{p+q-r}a_{s}a_{t}\cdot\|^2\Big)^{1/2}\\
    &\leq CN^{4\beta+1} (\cK + N) (\cN_{>N^{\alpha}/3}+1).
\end{split}
\]
The remaining contributions can be controlled in the same way, except the term $\Sigma_{3b}$: in this case, all momenta appearing in the creation operators are high and in order to efficiently use the kinetic energy, we bound this term in a more involved way by


\[
\begin{split}
    \Sigma_{3b}
    &\leq \!\!\!\sum_{\substack{r\in\Lambda^*_+, p,q, s, t\in\Lambda^*: p,q,s,t,\\ |q| > N^\beta, |s|, |t| \leq N^\beta,\\ |p|, |s+t-r|> N^\alpha/3}} \!\!\!   
    \frac{|s|+1}{|t|+1}\|(\cN_s+1)^{\frac12} a_p a_{s+t-r}\cdot\|
    \frac{|t|+1}{|s|+1}\|(\cN_q+1)^{\frac12} a_{p+q-r}a_{t}\cdot\|\\
    &\leq \Big(\!\!\!\sum_{\substack{r\in\Lambda^*_+, p,q, s, t\in\Lambda^*: p,q,s,t,\\ |q| > N^\beta, |p+q-r|, |s|, |t| \leq N^\beta,\\ |p|, |s+t-r|> N^\alpha/3}} \!\!\!   
    \Big(\frac{|s|+1}{|t|+1}\Big)^2
    \Big(\|\cN_s^{\frac12} a_p a_{s+t-r}\cdot\|^2 
    +\| a_p a_{s+t-r}\cdot\|^2\Big)
    \Big)^{1/2}\\
    &\quad\times \Big(\!\!\!\sum_{\substack{r\in\Lambda^*_+, p,q, s, t\in\Lambda^*: p,q,s,t,\\ |q| > N^\beta, |p+q-r|, |s|, |t| \leq N^\beta,\\ |p|, |s+t-r|> N^\alpha/3}} \!\!\!   
    \Big(\frac{|t|+1}{|s|+1}\Big)^2
    \Big(\|\cN_q^{\frac12} a_{p+q-r}a_{t}\cdot\|^2
    +\| a_{p+q-r}a_{t}\cdot\|^2
    \Big)
    \Big)^{1/2}\\
    &+\Big(\!\!\!\sum_{\substack{r\in\Lambda^*_+, p,q, s, t\in\Lambda^*: p,q,s,t,\\ |q|, |p+q-r| > N^\beta, |s|, |t| \leq N^\beta,\\ |p|, |s+t-r|> N^\alpha/3}} \!\!\!   
    \Big(\frac{|s|+1}{|t|+1}\Big)^2
    \Big(\|\cN_s^{\frac12} a_p a_{s+t-r}\cdot\|^2 
    +\| a_p a_{s+t-r}\cdot\|^2\Big)
    \Big)^{1/2}\\
    &\quad\times \Big(\!\!\!\sum_{\substack{r\in\Lambda^*_+, p,q, s, t\in\Lambda^*: p,q,s,t,\\ |q|,|p+q-r| > N^\beta, |s|, |t| \leq N^\beta,\\ |p|, |s+t-r|> N^\alpha/3}} \!\!\!   
    \Big(\frac{|t|+1}{|s|+1}\Big)^2
    \Big(\|\cN_q^{\frac12} a_{p+q-r}a_{t}\cdot\|^2
    +\| a_{p+q-r}a_{t}\cdot\|^2
    \Big)
    \Big)^{1/2}\\
    &\leq
    \big( N^{\frac32\alpha + \frac52\beta +\frac12}  (\cN_{>N^{\beta}}+N^{3\alpha})^{\frac12} \! +\! N^{3\alpha+\beta}   (\cN_{>N^{\beta}}+N^{3\alpha})\big) (\cK + \!N \!+\! N^{5\beta}) (\cN_{>N^{\alpha}/3}\!+\!1),
\end{split}
\]
where we set $\cN_s:=a_{s}^* a_{s}$.

Collecting the above estimates and multiplying by a factor $N^{2\kappa-2\alpha-2}$, we arrive at  
        \[\begin{split}
        &N^{2\kappa-2\alpha-2} \!\!\!\!\!\!\!\!\!\sum_{\substack{r\in\Lambda^*_+, p,q, s, t\in\Lambda^*:\\ (p+q-r,r),(s+t-r,r)\in (\PL^2)^c,\\ (p,q),(s,t)\in\PL^2 }} \!\!\!   |\langle \xi, a_p^* a_q^*  a_{s+t-r}^* a_{p+q-r}a_{s}a_{t}\xi\rangle|\\
        &\leq CN^{2\kappa-2\alpha-2}
        \langle\xi,\big(N^{4\alpha} (\cN_{>N^{\beta}}+N^{3\alpha})
        + N^{\frac52\alpha + \frac32\beta +\frac12} (\cN_{>N^{\beta}}+N^{3\alpha})^{\frac12}
        +N^{\frac32\alpha+\frac52\beta+1}\big)\\
        &\hspace{1cm}\times(\cK + N + N^{5\beta}) (\cN_{>N^{\alpha}/3}+1)\xi\rangle. \qedhere
        \end{split}\]
\end{proof}
Before concluding Theorem \ref{thm:main}, the last ingredient that we need is some mild a priori information on the energy of the ground state vector $\psi_N$, as remarked around Eq. \eqref{eq:markov}.
\begin{lemma}\label{lm:apriori}
    Let $ \psi_N$ denote the normalized ground state vector of $H_N$, defined in \eqref{eq:HN}, and let $  \beta \geq 0$. Then $ \psi_N$ satisfies the a priori bounds 
           \[\begin{split}    
           \langle \psi_N, \cN_{> N^{\beta}}\psi_N\rangle &\leq C  N^{1+\kappa-2\beta}, \\
           N^{-1}\langle \psi_N, \cK \cN_{> N^{\beta}}\psi_N\rangle &\leq C N^{1+2\kappa-2\beta} + C N^{\frac32\kappa+\frac12\beta}, \\
           N^{-2}\langle \psi_N, \cK \cN^2_{> N^{\beta}}\psi_N\rangle &\leq C N^{1+3\kappa-4\beta} + CN^{\beta+2\kappa} +CN^{\frac52\kappa-\frac32\beta}.
           \end{split}\]   
\end{lemma}
\begin{proof}  The first bound is a direct consequence of \eqref{eq:markov} and the fact that $ E_N \leq CN^{1+\kappa}$. For the bound on $\cK \cN_{\geq N^{\beta}}$, we use a commutator argument as in \cite{BBCS2, BBCS3, ABS}. We bound
         \[\begin{split}
        N^{-1}\langle \psi_N, \cK \cN_{> N^{\beta}}\psi_N\rangle & \leq    N^{-1}\langle \psi_N, \cK \psi_N\rangle^{\frac12}\langle \psi_N, \cN_{> N^{\beta}} \cK \cN_{> N^{\beta}}\psi_N\rangle^{\frac12}\\
        &\leq C N^{\frac{\kappa}2-\frac12} \langle \psi_N, \cN_{> N^{\beta}} \cK \cN_{> N^{\beta}}\psi_N\rangle^{\frac12}
        \end{split}\]
and then        
        \[\begin{split}
        \frac1N\langle \psi_N, \cN_{> N^{\beta}}\cK \cN_{> N^{\beta}}\psi_N\rangle & \leq  \frac1N\langle \psi_N,   \cN_{> N^{\beta}}  H_N \cN_{> N^{\beta}}\psi_N\rangle \\
        & = \frac{E_N}N \langle   \psi_N,  \cN^2_{> N^{\beta}}\psi_N\rangle + \frac1N\langle \psi_N,   \cN_{> N^{\beta}}\big[ H_N,  \cN_{> N^{\beta}}\big]\psi_N\rangle \\
        & \leq CN^{\kappa-2\beta}\langle \psi_N, \cK \cN_{> N^{\beta}}\psi_N\rangle+ \frac1N\langle \psi_N,   \cN_{> N^{\beta}}\big[ H_N,  \cN_{> N^{\beta}}\big]\psi_N\rangle.
        \end{split}\]     
To estimate the commutator contribution on the \rhs in the previous equation, we write 
        \[H_N -\cK=   \frac12\int_{\Lambda^2}dxdy\, N^{2-2\kappa}V(N^{1-\kappa}(x-y))\check{a}_x^*\check{a}_y^*\check{a}_x\check{a}_y=: \cV_N,\]
where $ \check{a}_x :=\sum_{p\in\Lambda^*} e^{ipx}a_p $ denotes the usual operator valued distribution annihilating a particle at $x\in \Lambda$, and we note $ [\cK, \cN_{> N^{\beta}}]=0$ as well as $[\cV_N, \cN_{> N^{\beta}}]=[\cN_{\leq  N^{\beta}}, \cV_N] $ with
        \[\begin{split} 
        \big[ \cN_{\leq  N^{\beta}}, \cV_N\big] = \sum_{p\in\Lambda^*: |p|\leq N^{\beta}} \int_{\Lambda^2}dxdy\, N^{2-2\kappa}V(N^{1-\kappa}(x-y))e^{ipx}\check{a}_p^*\check{a}_y^*\check{a}_x\check{a}_y + \hc
        \end{split}\]
Now, basic Cauchy-Schwarz estimates imply that 
        \[\begin{split}
             &N^{-1}\big|\big\langle \psi_N,   \cN_{> N^{\beta}}\big[ \cN_{\leq  N^{\beta}}, \cV_N\big]\psi_N\big\rangle\big|\\
             &\;\leq CN^{-1}\sum_{p\in\Lambda^*: |p|\leq N^{\beta}} \int_{\Lambda^2}dxdy\,N^{2-2\kappa}V(N^{1-\kappa}(x-y)) \|a_p \check{a}_y \,\cN_{> N^{\beta}} \psi_N\| \|\check{a}_x\check{a}_y\psi_N\| \\
             &\hspace{0.4cm}+CN^{-1}\sum_{p\in\Lambda^*: |p|\leq N^{\beta}} \int_{\Lambda^2}dxdy\,N^{2-2\kappa}V(N^{1-\kappa}(x-y)) \|a_p \check{a}_y  \psi_N\| \| \check{a}_x\check{a}_y \,\cN_{> N^{\beta}}\psi_N\|\\
             &\leq C N^{\frac\beta2+\frac\kappa2-1}\big( \|  (\cK+a^*_0a_0)^{\frac12}\cN_{> N^{\beta}}\psi_N \| \|\cV_N^{1/2}\psi_N\| +   \|  (\cK+a^*_0a_0)^{\frac12}\psi_N \| \|\cV_N^{1/2}\cN_{> N^{\beta}}\psi_N\|\big) \\
             &\leq C N^{\frac\beta2+ \kappa-\frac12}   \|  H_N^{\frac12}\cN_{> N^{\beta}}\psi_N \|+C N^{\frac\beta2+ \kappa}\|\cN_{> N^{\beta}}\psi_N\|.
        \end{split}\]
Combining the previous estimates with $ab\leq \frac{a^2}2+\frac{b^2}2$, we conclude   
        \[\begin{split}
        & \frac1N\langle \psi_N,   \cN_{> N^{\beta}}  H_N \cN_{> N^{\beta}}\psi_N\rangle  \\
         &\leq C N^{\kappa-2\beta}\langle \psi_N, \cK \cN_{> N^{\beta}}\psi_N\rangle + C N^{\frac\beta2+ \kappa-\frac12}   \|  H_N^{\frac12}\cN_{> N^{\beta}}\psi_N \|+C N^{\frac\beta2+ \kappa}\|\cN_{> N^{\beta}}\psi_N\|\\
         &\leq C N^{\kappa-2\beta}\langle \psi_N, \cK \cN_{> N^{\beta}}\psi_N\rangle + CN^{\beta+2\kappa} + \frac1{2N}\langle \psi_N,   \cN_{> N^{\beta}}  H_N \cN_{> N^{\beta}}\psi_N\rangle
        \end{split}\]
and therefore  
        \[\frac1N\langle \psi_N,   \cN_{> N^{\beta}}  H_N \cN_{> N^{\beta}}\psi_N\rangle\leq C N^{\kappa-2\beta}\langle \psi_N, \cK \cN_{> N^{\beta}}\psi_N\rangle+ CN^{\beta+2\kappa}.   \]       
As a consequence, we obtain that 
        \[\begin{split}
        N^{-1}\langle \psi_N, \cK \cN_{> N^{\beta}}\psi_N\rangle & \leq C N^{1+2\kappa-2\beta} + C N^{\frac32\kappa+\frac12\beta} , \\
         N^{-2}\langle \psi_N,   \cN_{> N^{\beta}}  \cK  \cN_{> N^{\beta}}\psi_N\rangle&\leq  C N^{1+3\kappa-4\beta} + CN^{\beta+2\kappa} +CN^{\frac52\kappa-\frac32\beta}. \qedhere
        \end{split}\]
\end{proof}

We are now ready to prove our main result. 
\begin{proof}[Proof of Theorem \ref{thm:main}]
Let $ \psi_N$ denote the normalized ground state vector of $ H_N$, given some parameter $ \kappa\in [0,\frac1{20})$. Let $\PL $ be defined as in \eqref{eq:defPL} and choose  
        $$\alpha := (1+\epsilon)\frac{41}{10}\kappa$$ 
for some sufficiently small $\epsilon>0$; in particular $ \alpha\in [0,1-\kappa]$. Now, by \eqref{eq:lwrbnd}, we have that 
        \[\begin{split}  \langle\psi_N, H_N\psi_N\rangle &\geq \sum_{r\in \Lambda_+^*}|r|^2 \langle \psi_N, c^*_r c_r \psi_N\rangle  - \langle \psi_N, R_N\psi_N\rangle \\
        &\hspace{0.4cm} +\frac{N^{\kappa}}{2N}\!\!\!\sum_{\substack{p,q,r\in\Lambda^*: \\ p,q, p+r, q-r \in \PL}}\!\!\!\!\! \langle \ph_{p+r}\otimes \ph_{q-r}, \vr \ph_p\otimes \ph_q  \rangle \langle \psi_N,  a^*_{p+r}a^*_{q-r}a_pa_q\psi_N\rangle
        \end{split}\] 
and our goal is to estimate the terms on the right hand side. We start with the kinetic energy term. Combining the bounds from Lemmas \ref{lm:K} and \ref{lm:apriori}, we find that
        \[\begin{split}
         &\sum_{r\in \Lambda_+^*}|r|^2 \langle \psi_N, c^*_r c_r \psi_N\rangle -4\pi^2 \langle\psi_N, \cN_+   \psi_N\rangle \\
         &\;\geq
         - C N^{1+\kappa-2 \delta}
         - CN^{1+3\kappa+\frac{1}{2}\delta-\frac{7}{2}\alpha} 
         - CN^{\frac{5}{2}\kappa+\frac{1}{2}\delta-\alpha} = o(N), 
        \end{split}\]
where we used \eqref{eq:markov}, the choice $\frac\kappa2 < \delta < \alpha$ and the identity $ \cN_{< N^{\delta}}- a^*_0a_0 = \cN_+-\cN_{\geq N^{\delta}}$.

Proceeding similarly for the remaining error terms, we obtain from Lemma \ref{lm:Vren} that 
        \[\begin{split} 
        &\frac{N^{\kappa}}{2N}\!\!\!\sum_{\substack{p,q,r\in\Lambda^*: \\ p,q, p+r, q-r \in \PL}}\!\!\!\!\! \langle \ph_{p+r}\otimes \ph_{q-r}, \vr \ph_p\otimes \ph_q  \rangle \langle \psi_N,  a^*_{p+r}a^*_{q-r}a_pa_q\psi_N\rangle \\
        &\;\geq 4\pi\mathfrak{a}N^{1+\kappa} - C  N^{1+2\kappa-2\alpha}  
        - CN^{\kappa+3\alpha} =4\pi\mathfrak{a}N^{1+\kappa}+o(N)
        \end{split}\]
and from Lemma \ref{lm:RN}, assuming $ \beta= (1+\epsilon)\frac52\kappa  $ for sufficiently small $\epsilon>0$, that 
        \[\begin{split} 
        |\langle \psi_N, R_N\psi_N\rangle| &\leq C N^{1+5\kappa-2\beta } + C N^{\frac12+\frac92\kappa+\frac52\alpha-2\beta } + C N^{1+\frac92\kappa+\frac12\beta-\frac32\alpha }\\
        &\hspace{0.5cm}+C N^{\frac12+4\kappa+\alpha+\frac12\beta }+C N^{1+4\kappa+\frac52\beta-\frac52\alpha } \\
        & = o(N) + C N^{\frac12+\frac52\alpha-\frac12\kappa+O(\epsilon) } + C N^{\frac12+\frac{21}4\kappa+\alpha+O(\epsilon)  } = o(N).
        \end{split}\]
Combining this with the ground state energy upper bound $ E_N\leq  4\pi \mathfrak{a}N^{1+\kappa}+o(N)$, as pointed out in the second remark after Theorem \ref{thm:main}, we get
        \[\begin{split}
           4\pi \mathfrak{a}N^{1+\kappa}+o(N) & \geq  \langle\psi_N, H_N\psi_N\rangle \geq 4\pi \mathfrak{a}N^{1+\kappa} + 4\pi^2 \langle\psi_N, \cN_+\psi_N\rangle +o(N)
        \end{split}\]
and thus conclude that
        \[\lim_{N\to\infty }N^{-1} \langle\psi_N, \cN_+\psi_N\rangle =\lim_{N\to\infty } \big(1- \langle\ph_0, \gamma_N^{(1)}\ph_0\rangle\big)=0. \qedhere \]         
\end{proof}

\vspace{0.2cm}
\noindent\textbf{Acknowledgements.} C.\ B.\ acknowledges support by the Deutsche Forschungsgemeinschaft (DFG, German Research Foundation) under Germany’s Excellence Strategy – GZ 2047/1, Projekt-ID 390685813. M.B., C.C and J.O. acknowledge partial support from the Swiss National Science Foundation through the Grant “Dynamical and energetic properties of Bose- Einstein condensates”, from the European Research Council through the ERC-AdG CLaQS, grant agreement n. 834782, and from the NCCR SwissMAP. C.C. acknowledges the GNFM Gruppo Nazionale per la Fisica Matematica - INDAM.


\appendix 

\section{Proof of the Bounds (\ref{eq:bndvr}) and (\ref{eq:bndeta})}\label{appx}
\begin{proof}[Proof of the Bounds (\ref{eq:bndvr}) and (\ref{eq:bndeta})] Throughout this appendix, we assume $p,q,s,t\in \Lambda^* $ and we abbreviate $\langle  \cL   \rangle_{pq, st}:=\langle\ph_{p}\otimes \ph_{q}, \cL \,\ph_{s}\otimes \ph_{t} \rangle  $
for every operator $ \cL $ on $L^2(\Lambda^2)$. Let us begin with a few elementary observations: it is clear that the operator $ N^{1-\kappa}V_N$ preserves the total momentum and that  
        \[\langle   N^{1-\kappa}V_N  \rangle_{pq,st}\leq \delta_{p+q, s+t}\int_{\Lambda^2} dx_1dx_2\, N^{3-3\kappa} V(N^{1-\kappa}(x_1-x_2))\leq  \delta_{p+q, s+t}\|V\|_1. \]
Combining this with the fact that $ -\Delta_{x_1}-\Delta_{x_2}+V_N$ and hence its pseudo-inverse 
        \[ \cR:=\pih \big[ \pih( -\Delta_{x_1}-\Delta_{x_2}+V_N)\pih \big]^{-1} \pih \]
from $ \pih L^{2}(\Lambda^2)$ to $\pih L^{2}(\Lambda^2)$ also preserve the total momentum, we get that 
        \[\begin{split}
            \langle  \vr \rangle_{pq,st} &\leq C\delta_{pq,st}\sup_{p,q\in\Lambda^*}  \Big( 1+N^{1-\kappa}\big\langle V_N^{1/2} \pih V_N^{1/2}  \cR V_N^{1/2}\pih  V_N^{1/2} \big\rangle_{pq,pq} \\
            &\hspace{3cm} + N^{1-\kappa}\big\langle V_N^{1/2} \pil V_N^{1/2}  \cR V_N^{1/2}\pil  V_N^{1/2} \big\rangle_{pq,pq}\Big).
        \end{split} \]
To control the right hand side, we make use of the operator inequalities 
        \[ 0\leq \pih V_N^{1/2}  \cR V_N^{1/2}\pih \leq 1 \;\text{ and }\; \cR \leq N^{-2\alpha}\pih \leq N^{-2\alpha}.  \]
This implies on the one hand that 
        \[ N^{1-\kappa}\big\langle V_N^{1/2} \pih V_N^{1/2}  \cR V_N^{1/2}\pih  V_N^{1/2} \big\rangle_{pq,pq}\leq  \langle N^{1-\kappa}V_N  \rangle_{pq,pq} = \|V\|_1 \]
and on the other hand that 
        \[ N^{1-\kappa}\big\langle V_N^{1/2} \pil V_N^{1/2}  \cR V_N^{1/2}\pil  V_N^{1/2} \big\rangle_{pq,pq} \leq  N^{-2\alpha} \|V\|_1 \| \pil  V_N^{1/2}\ph_p\otimes \ph_q \|^2_{\infty}.\]
Looking at the Fourier expansion
        \[  \pil  V_N^{1/2}\ph_p\otimes \ph_q = \sum_{s\in \PL:p+q-s\in\PL} N^{2\kappa-2} \widehat{V^{1/2}}((s-p)/N^{1-\kappa})\ph_{s}\otimes \ph_{p+q-s}, \]
the assumptions that $V\in L^1(\bR^3)$ having compact support and that $\alpha \leq 1-\kappa$ imply that $ \| \pil  V_N^{1/2}\ph_p\otimes \ph_q \|_{\infty}\leq CN^{\alpha}$ so that altogether $ |\langle   \vr \rangle_{pq,st}|\leq C \delta_{pq,st}$. 

Next, let us switch to the second bound in \eqref{eq:bndvr}. We first show that 
        \be\label{eq:app1} |\langle \vr\rangle_{00,00}-8\pi \mathfrak{a}|\leq CN^{\alpha+\kappa-1}.  \ee
Up to minor modifications, this bound follows as in \cite[Appendix A]{Br}, so let us focus on the key steps. Denote by $ f$ the zero energy scattering solution in $\bR^3$ such that  
        \[ (-2\Delta + V)f=0 \]
with $\lim_{|x|\to\infty } f(x) = 1$. It is well known (see \eg \cite[Appendix C]{LSSY}) that $ 0\leq f\leq 1$, that $f$ is radial and that for $x\in\bR^3$ outside the support of $V$, we have that $ f(x) = 1-\mathfrak{a} /|x|. $ Moreover, a basic integration by parts shows that 
        \[ 8\pi \mathfrak{a} = \int_{\bR^3} dx\, V(x)f(x) = \int_{\Lambda}dx\, N^{3-3\kappa}(Vf)(N^{1-\kappa}x) . \]
Let us denote $w:= 1- f$ which is easily seen to satisfy the bounds 
        $$ w(x) \leq \frac{C}{|x|}, \hspace{0.5cm} |\widehat w (p)|\leq \frac{C}{1+|p|^2} $$ 
for some constant $C>0$ (\eg based on the identity $ w = (-2\Delta)^{-1} Vf $). Moreover, pick a smooth bump function $ \chi \in C^{\infty}_c( B_{1/2}(0))$ such that $ \chi(x) =1$ if $ |x|\leq \frac14$ and define 
        \[(x_1,x_2) \mapsto \phi_N(x_1-x_2):= \chi(x_1-x_2)w(N^{1-\kappa}(x_1-x_2)) \in L^2(\Lambda^2).\]
By slight abuse of notation, we identify $\phi_N$ with the associated multiplication operator in $ L^2(\Lambda^2)$. As explained in \cite{Br}, we then have the identity 
        \[ \cR V_N \ph_0\otimes\ph_0
        = \phi_N \ph_0\otimes\ph_0
        + \cR \zeta_N \ph_0\otimes\ph_0
        - (1-\cR V_N)\pil \phi_N\ph_0\otimes\ph_0,  \]
where $ \zeta_N(x_1-x_2):=N^{\kappa-1}\zeta (x_1-x_2)$ for
        \[x\mapsto \zeta (x) :=2\mathfrak{a}\frac{ (\Delta\chi)(x)}{|x|} -4\mathfrak{a}\frac{  (\nabla \chi)(x)\cdot  (x)}{|x|^3}\in C^{\infty}_{0}\big(B_{1/2}(0)\cap \overline B^c_{1/4}(0)\big). \]
Using that $ 8\pi\mathfrak{a} = N^{1-\kappa}\langle V_N, (1-\phi_N)\rangle_{00,00}$, this yields
        \[ \langle \vr\rangle_{00,00}= 8\pi \mathfrak{a} + N^{1-\kappa}\langle \cR V_N, \zeta_N\rangle_{00,00}  + N^{1-\kappa}\langle V_N, (1-\cR V_N)\pil \phi_N\rangle_{00,00}.  \]
Now observe that for $ |p|>N^{\alpha}$, we have that
    \be\label{eq:appaux}\begin{split} \langle \cR V_N\rangle_{-pp,00}&= \frac{-\langle V_N\cR V_N\rangle_{-pp,00}}{2|p|^{2}} +\frac{\langle V_N-\pil(1-V_N\cR) V_N\rangle_{-pp,00}}{2|p|^{2}}\\
    &= \frac{\langle V_N(1-\cR V_N)\rangle_{-pp,00}}{2|p|^{2}}\end{split}\ee
and otherwise $\langle \cR V_N\rangle_{-pp,00}=0$ (by definition of $\cR$) \st $ N^{1-\kappa}\langle \cR V_N, \zeta_N\rangle_{00,00}\leq C N^{\kappa-1}$. Similarly, $ |\widehat\phi_N(p)|\leq CN^{\kappa-1}(1+|p|^2)^{-1}$ and Cauchy-Schwarz imply that 
        \[\begin{split}
          &| N^{1-\kappa}\langle V_N, (1-\cR V_N)\pil \phi_N\rangle_{00,00} |\\
          &\;\leq C \|V\|_1\big( \|\pil \phi_N\|_\infty+N^{-\alpha}\|\pil V_N^{1/2}\pil \phi_N\|_\infty\big)\leq C N^{\alpha+\kappa-1}.
        \end{split}  \]
Combining the previous estimates yields \eqref{eq:app1}. 
In fact, using $N^{1-\kappa}\langle V_N(1-\phi_N)\rangle_{p+r\, q-r,pq} = \widehat{Vf}(r/N^{1-\kappa})$ we can also compute $\langle \vr\rangle_{00,00}$ to higher precision and obtain that
\[
\begin{split}
\langle \vr\rangle_{00,00}
= &8\pi\mathfrak{a}
+\frac{N^{\kappa-1}}{2}\sum_{\substack{0\neq s\in\PL}}
\frac{(8\pi\mathfrak{a})^2}{|s|^2}     
+O(N^{\kappa-1})
+O(N^{2\kappa+2\alpha-2}).
\end{split}
\]
We omit the details as the second term is irrelevant for our range of $\kappa$, it only becomes relevant if one wants to consider the Lee-Huang-Yang order.

To get \eqref{eq:bndvr}, we combine \eqref{eq:app1} with two further steps. On the one hand, we have that 
        \be\label{eq:app2} |\langle  \vr \rangle_{pq,st}- \langle  \vr \rangle_{(p+q)0,(s+t)0}|\leq CN^{\kappa-1} \big( |p|+ |q|+|s|+|t|\big)  \ee
whenever $ p+q = s+t $. This bound follows very similarly as the first bound in \eqref{eq:bndvr}: since $ V_N$ is a multiplication operator, \eqref{eq:app2} clearly holds if we replace $ \vr$ by $ N^{1-\kappa}V_N$. Hence, it is enough to prove \eqref{eq:app2} for $ \vr$ replaced by $ N^{1-\kappa}V_N \cR V_N $. In this case, we write 
        \[\begin{split}
            &N^{1-\kappa}\langle V_N \cR V_N\rangle_{pq,st}- N^{1-\kappa}\langle  V_N \cR V_N\rangle_{(p+q)0,(s+t)0}\\
            &\;=  N^{1-\kappa} \langle \ph_p\otimes\ph_q, V_N \cR V_N  ( \ph_{s}\otimes\ph_{t} - \ph_{s+t}\otimes\ph_{0} )  \rangle  \\
            &\hspace{0.4cm} +  N^{1-\kappa}\langle ( \ph_{p}\otimes\ph_{q} - \ph_{p+q}\otimes\ph_{0} ),  V_N \cR V_N   \ph_{s+t}\otimes\ph_{0}   \rangle.
        \end{split}\]
Now, given any pair $ k,l \in \Lambda^*$, a direct computation shows that 
        \[\begin{split}
          N^{1-\kappa} \| V_N^{1/2}( \ph_{k}\otimes\ph_{l} - \ph_{k+l}\otimes\ph_{0} )  \|^2&= 2 \wh V(0) - 2 \wh V(l/N^{1-\kappa})\leq\frac{  C |l|^2 }{N^{2-2\kappa} }.  
        \end{split} \]
Note that the last step follows from a second order Taylor expansion and the fact that $ (\nabla_p \wh V(./N^{1-\kappa} )(0) =  N^{2-2\kappa} \int_{\bR^3} dx\, (-ix)  V(x) =0 $, $V$ being radial. Similarly, we get  
        \[\begin{split}
        & N^{-2\alpha}  \| \pil  V_N^{1/2}( \ph_{k}\otimes\ph_{l} - \ph_{k+l}\otimes\ph_{0} )
        \|^2_{\infty} \\
        &= N^{-2\alpha} \Big\|\sum_{\substack{ s\in \PL:\\k+l-s\in\PL}} \!\!\!\!\!\!N^{2\kappa-2}\big( \widehat{V^{1/2}}((s-k)/N^{1-\kappa}) - \widehat{V^{1/2}}((s-k-l)/N^{1-\kappa})\big) \ph_{s}\otimes \ph_{k+l-s}\Big\|_\infty^2\\
        &\leq C|l|^2/N^{2-2\kappa}. 
        \end{split}\]
Proceeding now as in the proof of the first bound in \eqref{eq:bndvr}, we obtain \eqref{eq:app2}. 

Combining \eqref{eq:app2} with \eqref{eq:app1}, the second bound in \eqref{eq:bndvr} thus follows if we prove that  
        \be\label{eq:app3}   |\langle V_N\cR V_N \rangle_{p0,p0}- \langle  V_N\cR V_N \rangle_{00,00}|\leq CN^{2\kappa-\alpha-2}  |p|^2  \ee
for every $ p\in\Lambda^*_+$. This can be proved similarly as detailed in \cite[Appendix A]{Br}: define 
        $$ -\Delta^{(p)}:= (-i\nabla_{x_1}+p)^2 -\Delta_{x_2}, \hspace{0.5cm} \cR^{(p)}:= \pih^{+} \big[\pih^{+} (-\Delta^{(p)}+V_N)\pih^{+}\big]^{-1}\pih^{+}, $$ 
where the orthogonal projection $\pih^{+}$ maps onto
        \[ \overline{\text{span}\{ \ph_k\otimes\ph_l: (k,l)\in (\PL^2)^c \text{ and }l\neq 0  \} }.\]
Notice that this ensures $\langle \xi, -\Delta^{(p)}\xi\rangle \geq 4\pi^2 $ for every $ \xi\in \pih^{+}L^2(\Lambda^2)$, by construction of the projection $\pih^{+} $. In particular, $\cR^{(p)}$ is well defined. Now, based on the observation 
        \[\langle V_N\cR V_N \rangle_{p0,p0} = \langle V_N e^{-ipx_1}\cR e^{ipx_1} V_N  \rangle_{00,00} \] 
and the fact that $ -\Delta^{(p)} = e^{-ipx_1}(-\Delta_{x_1}-\Delta_{x_2})e^{ipx_1} $ for $ p\in\Lambda^*$, it follows that
        \[\begin{split}
         &\langle V_N e^{-ipx_1}\cR e^{ipx_1} V_N\rangle_{00,00}-\langle V_N\cR^{(p)}V_N\rangle_{00,00} \\
         & = -\langle V_N \cR^{(p)}e^{-ipx_1}\pil(1-V_N\cR)V_N\rangle_{00,p0} +\langle V_N (1-\cR^{(p)}V_N)\pil^{+} e^{-ipx_1}\cR  V_N\rangle_{00,p0}.
        \end{split}\]
Here, we set $\pil^{+}:= 1- \pih^{+}$. Since $V_N \cR^{(p)} $ preserves the total momentum and projects onto a subset of $(\PL^2)^c $, we have that 
        \[\begin{split}
            &\langle V_N \cR^{(p)}e^{-ipx_1}\pil(1-V_N\cR)V_N\rangle_{00,p0} \\ 
            &= \sum_{q\in \PL^c\cap \Lambda^*_+; \,s,t\in \PL} \langle V_N \cR^{(p)}\rangle_{00,-qq} \langle \ph_{-q}\otimes \ph_q, \ph_{s-p}\otimes \ph_t\rangle\langle (1-V_N\cR)V_N\rangle_{st,p0}=0.  
        \end{split}\]
On the other hand, using that $ \cR = \pih \cR$ so that  
        \[\begin{split}
             \langle \cR V_N\rangle_{(p-q)q, p0} & = \frac{\langle V_N- V_N\cR V_N\rangle_{(p-q)q, p0}}{|p-q|^2 + |q|^2}-\frac{\langle \pil (1-V_N\cR )V_N\rangle_{(p-q)q, p0}}{|p-q|^2 + |q|^2}  \\
             &= \frac{\langle V_N- V_N\cR V_N\rangle_{(p-q)q, p0}}{|p-q|^2 + |q|^2}
        \end{split}  \]
if $(p-q,q)\in (\PL^2)^c $ and $ \langle \cR V_N\rangle_{(p-q)q, p0} =0$ otherwise, we obtain that
        \[\begin{split}
            \big\| \pil^{+} e^{-ipx_1} \cR  V_N e^{ipx_1}\big\|_{\infty} &\leq   \sum_{s\in  \PL} | \langle \cR V_N\rangle_{(p-s)s, p0}|\leq CN^{\alpha+\kappa-1},\\ 
             \big\| \pil^{+} V_N^{1/2}\pil^{+} e^{-ipx_1} \cR  V_N e^{ipx_1}\big\|_{\infty} &\leq   \sum_{s,t\in \PL} \frac{C}{N^{3-3\kappa}(1+|t|^2)} \leq CN^{2\alpha+\kappa-1}.
        \end{split}\] 
Hence, arguing similarly as in the previous steps, we find that
        \[\begin{split}
            &|\langle V_N (1-\cR^{(p)}V_N)\pil^{+} e^{-ipx_1}\cR  V_N\rangle_{00,p0}|\\ &\leq \frac{CN^{\kappa}}{N}\big(   \big\| \pil^{+} e^{-ipx_1} \cR  V_N e^{ipx_1}\big\|_{\infty}+N^{-\alpha}\big\|\pil^{+} V_N^{1/2} \pil^{+} e^{-ipx_1} \cR  V_N e^{ipx_1}\big\|_{\infty}  \big) \leq C N^{\alpha+2\kappa-2}.
        \end{split}\] 
Notice here that we used additionally the operator inequalities $ -\Delta^{(p)}   \geq N^{2\alpha}\pih^+ $ and, as a consequence, $ \cR^{(p)} \leq N^{-2\alpha}\pih^+\leq N^{-2\alpha}$ in the image 
        $$ \pih^+ \big(\textbf{1}_{\text{P} = 0}\,L^2(\Lambda^2)\big) =\pih^{+}\, \overline{\text{span}\{ \ph_s\otimes \ph_{-s}:s\in\Lambda^*\}} =  \overline{\text{span}\{ \ph_s\otimes \ph_{-s}:|s| > N^{\alpha}\}} $$ 
of the space of zero total momentum $\text{P}:=-i\nabla_{x_1}-i\nabla_{x_2}$ under $\pih^+$, and that both 
        \[  V_N^{1/2}\pil^{+}V_N^{1/2} \pil^{+} e^{-ipx_1}\cR  V_N e^{ipx_1}\in L^2(\Lambda^2)\;\; \text{ and } \;\;V_N^{1/2}\pil^{+}V_N^{1/2}\in L^2(\Lambda^2) \]
have zero total momentum. 

Collecting the previous bounds, proving \eqref{eq:app3} reduces to proving that 
        \[|\langle V_N\cR^{(p)} V_N \rangle_{00,00}- \langle  V_N\cR V_N \rangle_{00,00}|\leq CN^{2\kappa-\alpha-2}  |p|^2.\]
To show this, we use that 
        \[ | \langle\cR^{(sp)}V_N \rangle_{-qq,00}| =  \frac{|\langle V_N-V_N\cR^{(sp)}V_N\rangle_{-qq,00}|}{|sp-q|^2+|q|^2}\leq \frac{CN^{\kappa-1}}{|sp-q|^2+|q|^2}  \]
for all $s\in [0,1]$ and $|q|>N^{\alpha}$ (otherwise $\langle\cR^{(sp)}V_N \rangle_{-qq,00}=0$). Together with 
        $$\langle V_N\cR^{(p)} V_N \rangle_{00,00} = \langle V_N\cR^{(-p)} V_N \rangle_{00,00} $$ 
and a second order Taylor expansion, we find that 
        \[|\langle V_N\cR^{(p)} V_N \rangle_{00,00}- \langle  V_N\cR V_N \rangle_{00,00}|\leq CN^{2\kappa-2}  |p|^2\int_{|q|> N^{\alpha}}dq\,  |q|^{-4}\leq N^{2\kappa-\alpha-2} |p|^2, \]
which implies \eqref{eq:app3} and thus \eqref{eq:bndvr}.

Finally, \eqref{eq:bndeta} is a direct consequence of the identity $(-\Delta_{x_1}-\Delta_{x_2}) \eta =  \pih  \vr \pil$ 
and the bound \eqref{eq:bndvr} implying that $( |p|^2+|q|^2) | \langle \eta \rangle_{pq,st}|  \leq C \delta_{pq,st} \textbf{1}_{( \PL^2)^c}((p,q))\textbf{1}_{ \PL^2}((s,t))$. \qedhere 
\end{proof}


\end{document}